\documentclass[12pt]{article}
\pdfoutput=1
\usepackage{amsmath,amssymb, braket, amsthm, bm}
\usepackage[a4paper]{geometry}
\usepackage[colorlinks=true, pdfstartview=FitV, linkcolor=blue, citecolor=blue, urlcolor=blue]{hyperref}
\usepackage{graphicx}
\usepackage{subcaption}
\usepackage{xcolor}

\newcommand{\ind}{\mbox{ ind}}
\def\ran{\mathop{\rm Ran}\nolimits}

\newcommand {\bC}{{\mathbb C}}
\newcommand {\bZ}{{\mathbb Z}}
\newcommand {\bN}{{\mathbb N}}
\newcommand {\bM}{{\mathbb M}}
\newcommand {\bR}{{\mathbb R}}
\newcommand {\bS}{{\mathbb S}}
\newcommand {\bH}{{\mathbb H}}
\newcommand {\bL}{{\mathbb L}}
\newcommand {\bI}{{\mathbb I}}

\newcommand {\cE}{{\mathcal E}}

\newcommand {\cK}{{\mathcal K}}
\newcommand {\cR}{{\mathcal R}}

\newcommand{\cO}{{\mathcal O}}
\newcommand {\cV}{{\mathcal V}}

\newcommand{\cshift}{{\bf T}}
\newcommand {\flux}{{\Phi}}
\newcommand {\bbM}{{\bf M}}
\newcommand {\bbC}{{\bf C}}
\newcommand {\bbP}{{\bf P}}

\newcommand{\be}{\begin{equation}}
\newcommand{\ee}{\end{equation}}

\newtheorem{thm}{Theorem} [section]
\newtheorem{lem}[thm]{Lemma}
\newtheorem{prop}[thm]{Proposition}
\newtheorem{proposition}[thm]{Proposition}

\newtheorem{definition}[thm]{Definition}

\newtheorem{cor}[thm]{Corollary}

\newtheorem {rem}[thm]{Remark}

\newtheorem {remark}[thm]{Remark}

\title{ Examples for stable quantum currents
\thanks{ Supported by
ECOS-Conicyt C15E10}}
\author{Joachim Asch, Mohamed Mouneime \thanks{CNRS, CPT, Aix Marseille Universit\'e, Universit\'e de Toulon, Marseille, France, asch@cpt.univ-mrs.fr}}
\date{12/8/19}
\begin{document}
\maketitle

\begin{abstract}

We provide a mathematical analysis of  two models advocated  in the theoretical and experimental condensed matter  literature: the two dimensional spin-$1/2$ Quantum Walk and the Kagome quantum network; they apply to occurrence of stable quantum currents. This illustrates the theory of stable absolutely continuous spectrum and stable  currents developed in \cite{ABJ5}.
 \end{abstract}

\section{Introduction}

Stable quantum currents of topological origin are of great interest from the points of view of mathematical physics as well as for engineering of topological materials.
Here we are in particular interested in the  two dimensional spin-$1/2$ Quantum Walk proposed in \cite{kitagawa} in order to model topological phases in two dimensions and to be technically feasible; it  was recently experimentally realized in a photonic setup \cite{chenEtAl}. On the other hand oriented quantum network models, originally proposed to quantitatively understand the Quantum Hall localization-delocalization transition \cite{cc,kok} continue to be a formidable mathematical challenge \cite{cardy}. In the present article we concentrate on the delocalization aspect of the Kagome network which leads to the study of the Ruby graph. This was used to describe arrays of coupled optical resonators \cite{pasekchong, dft}.

Mathematically speaking both models are a discrete quantum dynamically system characterized by a unitary operator $U$ on $\ell^2(\cV;\bC^d)$, $\cV$ being the vertices of a graph,  which is parametrized by a countable family of $U(2)$ matrices.  Depending on this family these models may have any dynamical behavior.  Known results  which concern delocalization cover the existence of a current and full absolutely continuous spectrum for the  L-lattice network of Chalker Coddington and   coined  Quantum Walks \cite{ABJ5}.  Further mathematical results assert  that the presence of boundaries and symmetries of the bulk imply occurrence of currents, \cite{rlbl,dft,graftauber,ssb}.

We apply the method of \cite{ABJ5}
concerning the topological properties of the self-adjoint flux operator out of the subspace $\ran(P)$ of an orthogonal projection
$$
\flux=U^*PU-P.
$$ 
In particular the existence of a wandering subspace is guaranteed by a non vanishing index of $\Phi$; the subspace  reduces a perturbation of the  evolution operator to  a shift. This may be considered as a current; its existence implies  gapless absolutely continuous spectrum at all quasienergies under decay assumptions. 
 
 In the next section we discuss the Kagome quantum network, exhibit how it leads to the Ruby graph and show under very general conditions that a non trivial flux  through a curve exists, see Theorem \ref{thm:main} and Remark \ref{rem:cap}.
 
 Then we present our analysis of the two dimensional spin-$1/2$ Quantum Walk and prove the occurrence of an edge state in Theorem \ref{thm:currentKwalk}, again under very mild conditions,  thus providing a mathematical theory for the experimental findings of \cite{chenEtAl}. 

To be self contained we present known results on the flux operator in an Appendix.

\section{Transport on  the Ruby graph}\label{section:cc}

We present  transport results for the Chalker-Coddington model defined on the full  two-dimensional Ruby graph which describes the Kagome quantum network. We employ the methods developed in \cite{ABJ5}, in particular the study of a  flux operator for a projection on a suitably chosen half space.

The original Chalker-Coddington model on the L-network  was designed for a numerical study of  quantum transition  events between microscopic currents,  for background see  \cite{cc,kok, cardyoverview,ABJ1}. 
In terms of the physical model  the currents are encoded by the vertices of the Ruby graph $\cR$ while  their microscopic transition events   are encoded by a collection of unitary $2 \times 2$ matrices 
\[\left\lbrace S_{z}\right\rbrace_{z\in\cK}\in U(2)\] 
indexed by the vertices of the   Kagome graph $\cK$.

The collection $\left\lbrace S_{z}\right\rbrace_{z\in\cK}$, which is a priori arbitrary, represent the parameters of the model. It  defines the unitary operator:
\[U:\ell^2(\cR;\bC)\to\ell^2(\cR;\bC)\]
and we shall study the flux operator $\Phi=U^\ast P U-P$ for $P$ a projection on a half space.

We now present the details of this model. 
Let $\lbrace v_1,v_2\rbrace$ be a positively oriented basis of $\bR^2$ and denote 
\[\cK_1:=\{z\in\bR^2; z= n_1 v_1+n_2 v_2,\quad n_1,n_1\in\bZ\}; \quad \cK_2:=\cK_1+\frac{1}{2}v_1; \quad \cK_3:=\cK_3+\frac{1}{2}v_2.\]
The set of the  Kagome vertices is \[\cK:=\cK_1\bigcup \cK_2\bigcup\cK_3.\]

\begin{figure}[hbt]
\centerline {
\includegraphics[width=8cm]{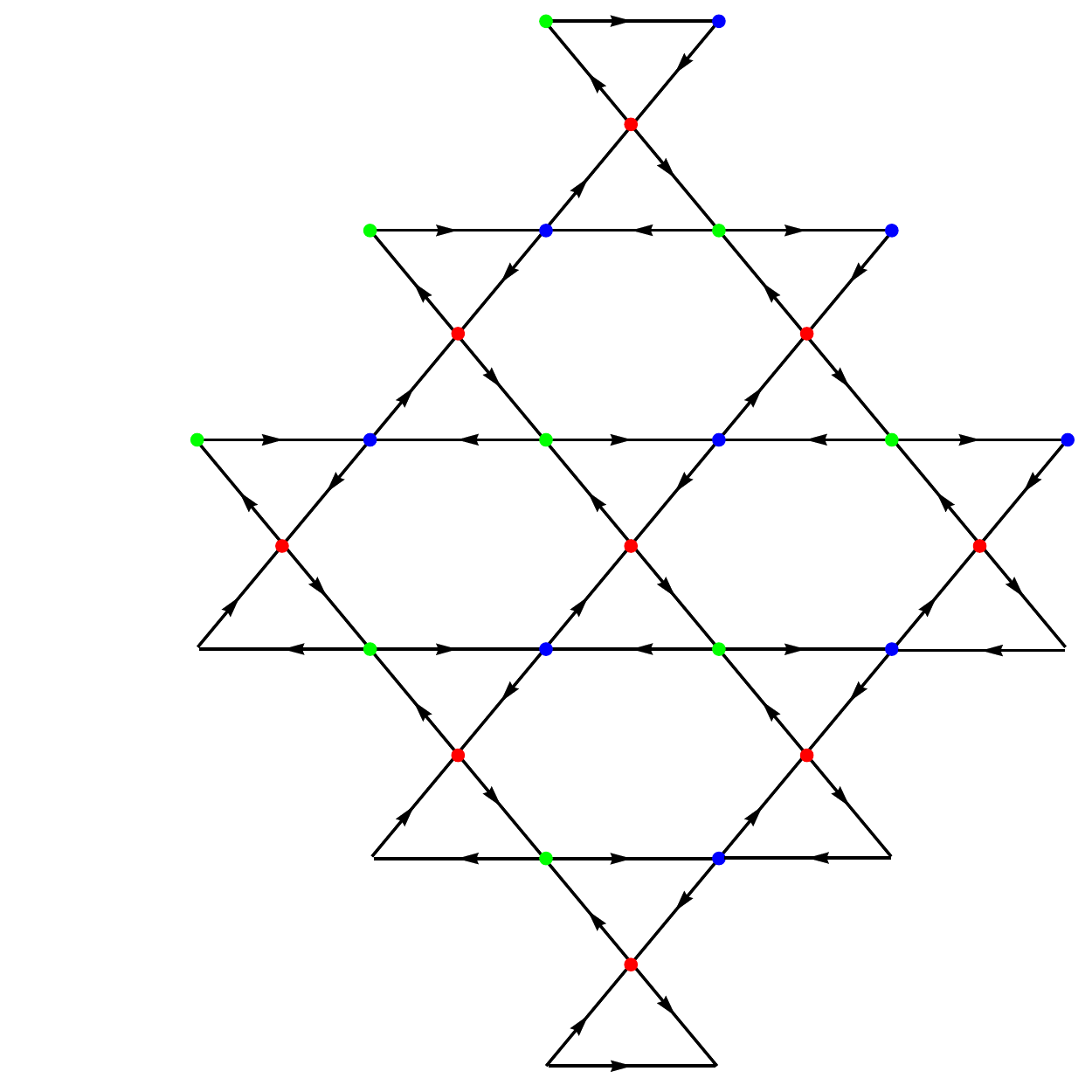}
}
\caption{The Kagome scattering network represented as graph with vertices $\cK$;  its edges connect  nearest neighbors, their orientation is anticlockwise around the hexagons and clockwise around the triangles.}
\label{fig:kagomegraph}
\end{figure}

The edges connect each vertex to its next neighbors. The faces adjacent to a vertex are in turn: a triangle, a hexagon, a triangle, a hexagon; the orientation of the edges is such that the orientation of the boundaries of the hexagons is anticlockwise, and clockwise for the triangles, see figure (\ref{fig:kagomegraph}) for the choice $v_1=e\left(\frac{\pi}{3}\right),v_2=e\left(\frac{2\pi}{3}\right)$. 

The physical model describes microscopic currents flowing along these edges. The two ingoing currents per vertex are scattered to the two outgoing currents which is encoded by the $2\times2$ unitary at the vertex. 
We label an edge (as segment of $\bR^2$) by its center. The set of edge-centers are the vertices $\cR$ of  the Ruby graph, {see figure (\ref{fig:rubygraph})}. For background on Archimedean Lattices see \cite{rsh}.  Explicitly:

\[ \cR=\{z\pm\frac{1}{4}v_j; j\in{1,2}, z\in\cK_1\}\bigcup\{z+\frac{1}{2}v_1\pm\frac{1}{4}(v_1-v_2); z\in\cK_1\}.\]

Then for $w\in\cR$ and  with the canonical base vectors \[\Ket{w}\in\ell^2(\cR;\bC) \hbox{ defined by } \ket{w}:= \left(\cR\ni x\mapsto\delta_{w,x}\in\bC\right),\] the edge with center $w$ is associated to the subspace $span\{\Ket{w}\}$; for a vertex $z$ of the Kagome graph define 
\[Q_z : \ell^2(\cR;\bC)\to \ell^2(\cR;\bC) \quad\hbox{ the projection on the  subspace  incoming to $z\in\cK$; }\]
as an example for $z\in\cK_1$: $Q_z=\Ket{z+\frac{1}{4}v_1}\Bra{z+\frac{1}{4}v_1}+\Ket{z-\frac{1}{4}v_1}\Bra{z-\frac{1}{4}v_1}$.
As  exactly two edges are incident to exactly one vertex we have
\[\ell^2(\cR;\bC)\equiv\bigoplus_{z\in\cK}\ran Q_z\]
and define the unitary $U$ uniquely by the collection of the $2\times2$-unitaries:
\[U:\ell^2(\cR;\bC)\to\ell^2(\cR;\bC); \qquad U\restriction\ran Q_z:=S_z\restriction\ran Q_z.\]
The subspace associated to the unique two outgoing edges at $z\in\cK$ is $\ran \widehat{Q}_z$ with the projection
\[\widehat{Q}_z:=U Q_zU^\ast,\]
see figure (\ref{fig:scatteringsubspaces}).

\begin{figure}[hbt]
\centerline {
\includegraphics[width=6cm]{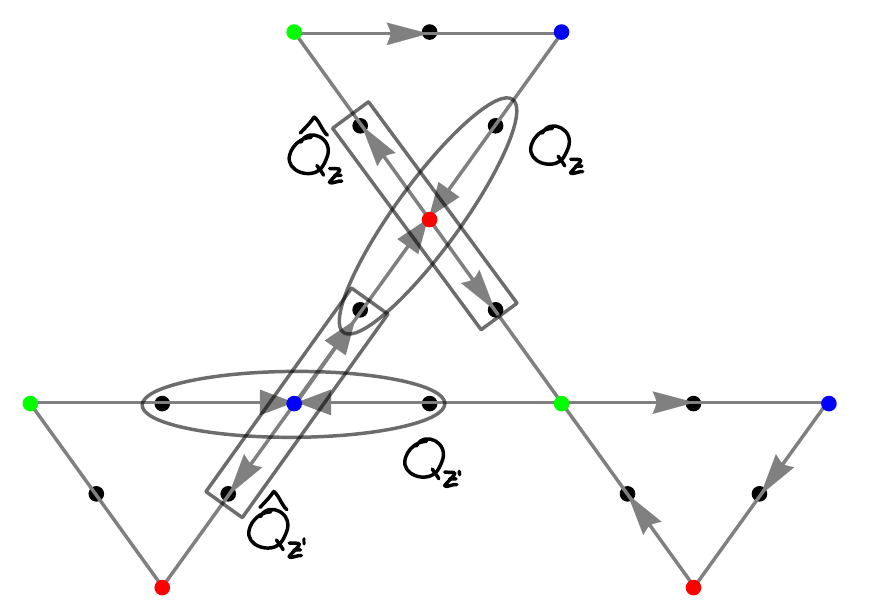}
}
\caption{The incoming and outgoing projections $Q_z$ and $\widehat Q_z$ }
\label{fig:scatteringsubspaces}
\end{figure}
\begin{rem} Note that $Q_z,\widehat{Q}_z$ do not depend on the collection of scattering matrices.
\end{rem}

The basic tool for the discussion of the flux operator is the following proposition which we adapt form \cite{ABJ5}, proposition 4.2, including its short intuitive proof.

\begin{prop}\label{prop:indexformula} Let $M$ be a subset of $\cR$ and $P$ the multiplication operator $P=\chi(x\in M)$, 
\[\Phi=U^\ast P U -P.\]
Then it holds

\begin{enumerate}
\item $\left\lbrack\Phi,Q_z\right\rbrack=0 \quad \forall z\in\cK$;
\item $\Phi Q_z=U^\ast (P\widehat{Q}_z)U- PQ_z$;
\item $\ind(\Phi)$ is well defined iff for a $c,R>0$ : $\sup_{\vert z\vert> R}\Vert \Phi Q_z\Vert\le c<1$;  then it holds:
\begin{equation}\label{eq:indexformula}
\ind(\Phi)=\sum_{z, \vert z\vert \le R} \dim\ran(P\widehat{Q}_z)-\dim\ran(PQ_z).
\end{equation}
\end{enumerate}
\end{prop}

\begin{proof} By definition $Q_z,\widehat{Q}_z$ are multiplication operators on $\ell^2(\cR;\bC)$ thus commute with $P$, or
\[\lbrack\Phi,Q_z\rbrack=\lbrack U^\ast P U,Q_z\rbrack=U^\ast\lbrack P, \widehat{Q}_z\rbrack U=0.\]
Also $\ker\left(\Phi^2-\bI\right)=\bigoplus_{z\in\cK}\ker\left((\Phi^2-\bI)Q_z\right)$ and $\sigma\left(\Phi^2\right)=\bigcup_{z\in\cK}\sigma\left(\Phi^2Q_z\right)$. By  assumption $dist\left(\{1\},\bigcup_{\vert z\vert>R}\sigma(\Phi^2Q_z)\right)>0$ thus $1$ is an isolated finite dimensional  eigenvalue and 
\[\ind(\Phi)=\sum_{\vert z\vert\le R} \dim\ker\left((\Phi-\bI)Q_z\right)-\sum_{\vert z\vert\le R}\dim\ker\left((\Phi+\bI)Q_z\right)=\]\
\[=\sum_{\vert z\vert\le R}\ind(\Phi Q_z)=\sum_{\vert z\vert\le R}trace(\Phi Q_z)=\sum_{z, \vert z\vert \le R} \dim\ran(P\widehat{Q}_z)-\dim\ran(PQ_z).\]
\end{proof}
In order to efficiently define the projection and to count the relevant dimensions we associate 
 to the family of unitaries $U$  the directed Ruby graph  $G=(\cR,\cE)$ whose set of vertices is $\cR$ and whose directed edges $\cE$
 are defined by, c.f. figure (\ref{fig:rubygraph}),
\[\overrightarrow{xy}\in \cE \hbox{ \it iff } \langle y,U_c\ x\rangle\neq0\]
where $U_c$ is the  model which has all its scattering matrices $S_z=\frac{1}{\sqrt{2}}\left(
\begin{array}{cc}
  1   &-1   \\
  1   &  1  
\end{array}
\right)$.

For a given $U$ we attribute
\[\hbox{ the weight }\vert\langle y,U x\rangle\vert \hbox{ to the edge } \overrightarrow{xy}\in \cE.\]
For a given $z\in\cK$ there are four edges, we label the weight by $\vert r_z\vert$ (color: black)  if the corresponding current  is scattered to its left and by $\vert t_z\vert$ if it is scattered to its right (color: red). Explicitly we have for $z\in\cK_1$
\[\left\vert\left\langle z-\frac{1}{4}v_2,U (z+\frac{1}{4}v_1)\right\rangle\right\vert=:\left\vert t_z\right\vert, \qquad \left\vert\left\langle z+\frac{1}{4}v_2,U (z+\frac{1}{4}v_1)\right\rangle\right\vert=:\vert r_z\vert\]
and all other cases analogously, see figure (\ref{fig:weights}). 

Remark that for a choice of a basis in each $\ran Q_z, \ran \widehat Q_z$ which is compatible with this convention, the parameters are
\[S_z\equiv q_z\left(
\begin{array}{cc}
  r_z   &-t_z   \\
  \bar{t_z}   &  \bar{r_z}  
\end{array}
\right)\quad \hbox{ with } q_z\in S^1, r_z,t_z\in\bC \hbox{ s.t. } \vert r_z\vert^2+\vert t_z\vert^2=1.\]

\begin{figure}[hbt]
\centerline {
\includegraphics[width=6cm]{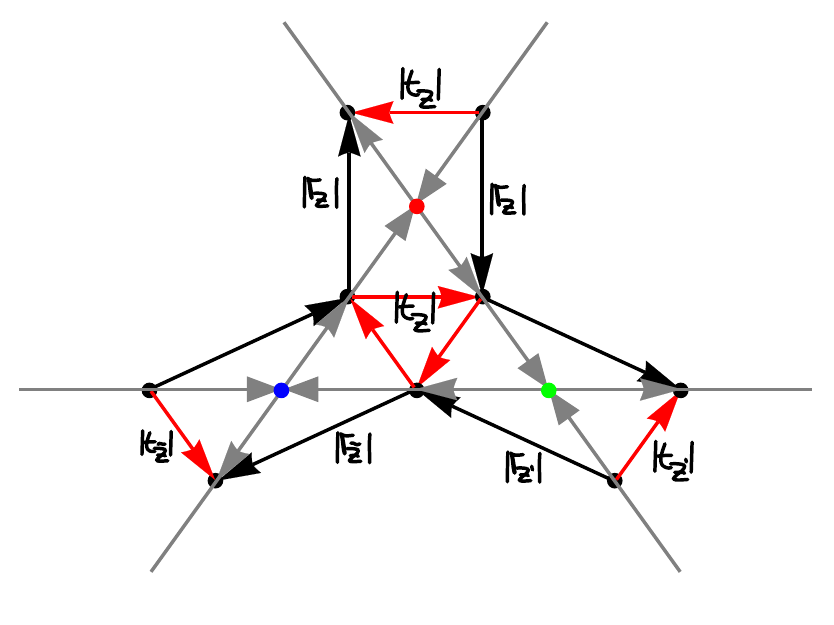}
}
\caption{The weights of $(\cR,\cE)$.}
\label{fig:weights}
\end{figure}

The projection $P$ will be defined by a path in the dual graph.

The faces of the Ruby graph figure (\ref{fig:rubygraph}) which are  adjacent to a fixed vertex are in turn  a triangle, a rectangle, a hexagon, a rectangle. The dual graph 
\[G^\ast=(\cR^\ast,\cE^\ast)\]
has its vertices at the center of the faces of $(\cR,\cE)$
and the dual edges (which we call links to distinguish) are segments which connect to all next neighbors  thus bisecting  edges of $\cE$. Remark that the scattering points are the centers of the rectangles in the Ruby graph:  $\cK\subset\cR^\ast$.

\begin{figure}[hbt]
\centerline {
\includegraphics[width=9cm]{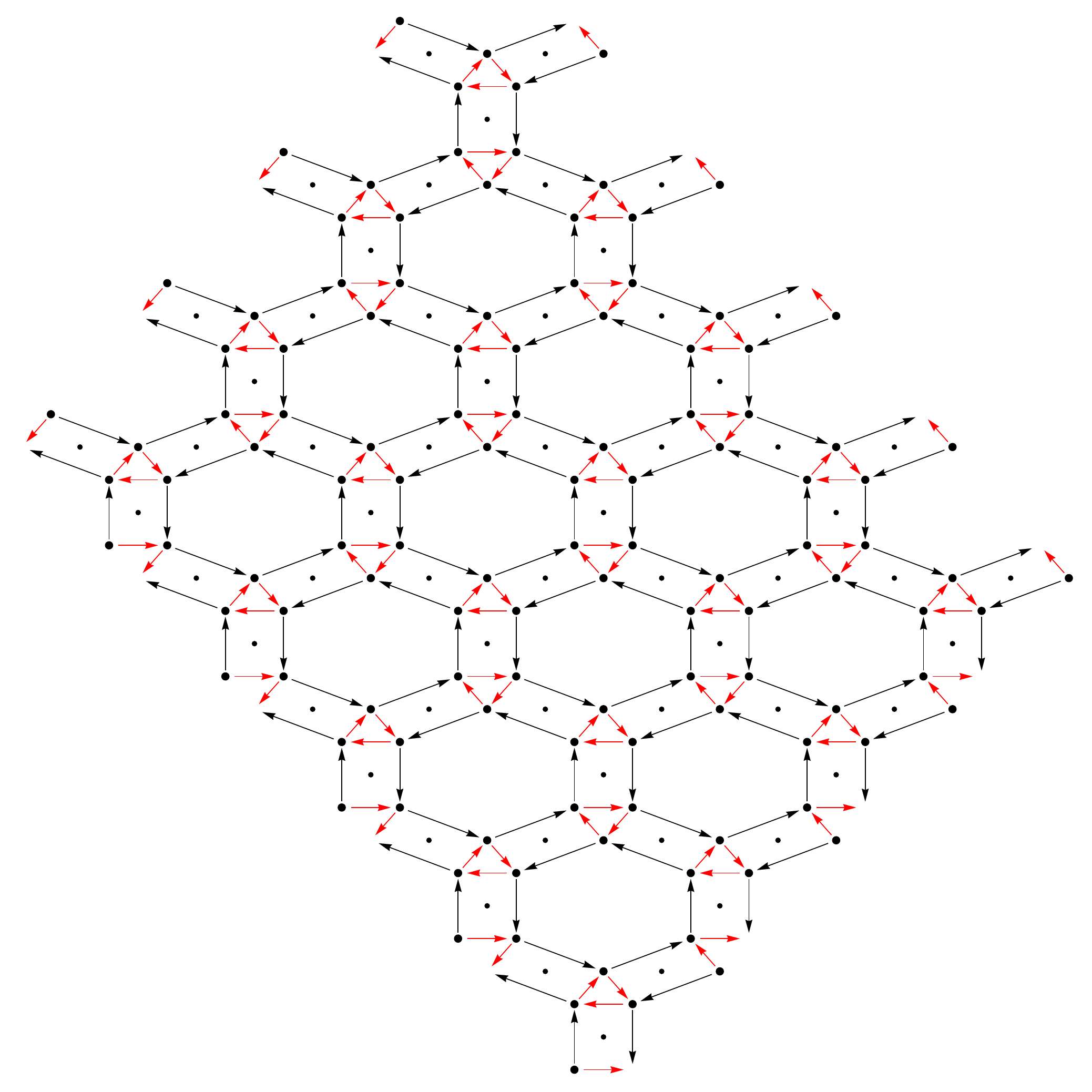}
}
\caption{The Ruby graph}
\label{fig:rubygraph}
\end{figure}

\begin{figure}[hbt]
\centerline {
\includegraphics[width=9cm]{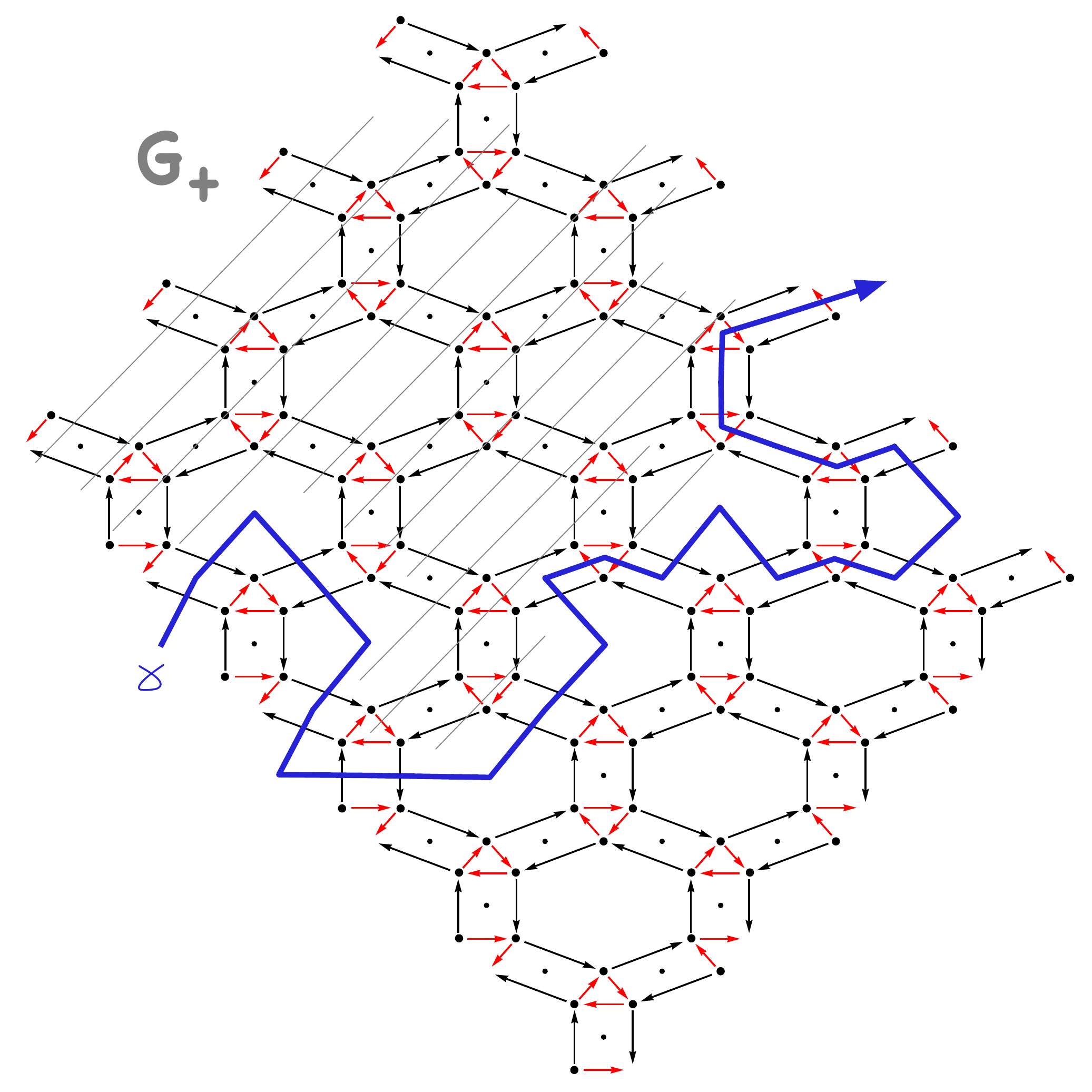}
}
\caption{An admissible path defining the projection on the states on its left}
\label{fig:projection}
\end{figure}

Now consider a partition of $G$ in two infinite connected subgraphs $G_+$ and $G_-$, consider a path in $G^\ast$ which bisects the edge boundary of $G_+$
and visualize this path as an injective, continuous, and piecewise unit speed curve of straight segments on the dual graph in $\bR^2$
\[\gamma: \bR\to G^\ast\subset\bR^2 \hbox{ such that for integer } t\in\bZ: \gamma(t)\in \cR^\ast \hbox{ and } \overrightarrow{\gamma(t)\gamma(t+1)}\in \cE^\ast\]
oriented such that $G_+$ is to the left of $\gamma$; we call it an admissible\ path, see figure (\ref{fig:projection}).
Let $\cR_+$ be the set of vertices of $G_+$  and consider its projection and flux operator in $\ell^2(\cR; \bC)$ :
\[P_\gamma=\chi(x\in \cR_+), \quad \Phi_\gamma=U^\ast P_\gamma U-P_\gamma.\]

In the following we shall give sufficient conditions for non trivial $\ind(\Phi_\gamma)$. By the general theory of Theorem \ref{thm:general} and Proposition \ref{lem:wandering} we know that all vertices which are not incident to the edge boundary of $G_+$ belong to  $\ker(\Phi_\gamma)$;  in view of the basic formula (\ref{eq:indexformula}) we  label the set of interesting vertices as follows:

\begin{definition}Let $\gamma:\bR \to G^\ast$ be an admissible path. We call $\cE_\gamma\subset \cE$
the set of edges bisected by $\gamma$ and $\cK_\gamma$ the labels of subspaces $\ran Q_z$ which contain  vertices incident to $\cE_\gamma$, i.e. :
\[\cK_\gamma:=\left\{z\in\cK; P_\gamma^\perp U_c P_\gamma Q_z\neq0 \hbox{ or } P_\gamma U_c P_\gamma^\perp Q_z\neq0\right\}.\]
\end{definition}

From this graphical representation of $U$ and $P$ we can literally read off how to construct a useful admissible path. The basis observation is:
\begin{rem}\label{rem:dimensions} From equation (\ref{eq:indexformula}) and the fact that $Q_z$ and $\widehat{Q}_z$ are diagonally opposite  vertices of the rectangles of the Ruby graph, we observe:  consecutively bisected  links of type hexagon-rectangle-hexagons as well as  triangle-rectangle-triangle do not contribute to the index because at the common dual vertex $z\in\cK$ it holds $\dim\ran PQ_z=\dim\ran P\widehat Q_z$. On the other hand consecutively bisected  links of type hexagon-rectangle-triangle as well as  triangle-rectangle-hexagon do  contribute to the index because $\dim\ran PQ_z\neq\dim\ran P\widehat Q_z$, see figure (\ref{fig:randtpath}).
\end{rem}
We define the non contributing parts :

\begin{definition}Let $I\subset\bR$. An admissible path $\gamma$ is called  a $r$-path in $I\subset\bR$ if its restriction to $I$ bisects only edges of weight $\vert r\vert$ and a $t$-path in $I\subset\bR$ if its restriction to $I$ bisects only edges of weight $\vert t\vert$.
\end{definition}

Concerning the regularity of $\Phi_\gamma$ we have
\begin{figure}
    \centering
    \begin{subfigure}[b]{0.20\textwidth}
        \includegraphics[width=\textwidth]{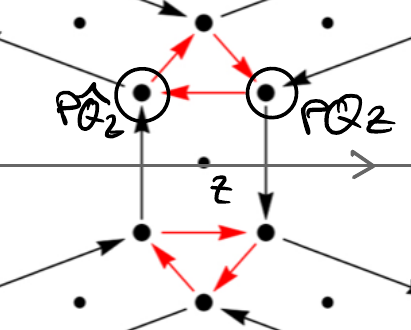}
 \caption{r and r edge}
       \label{fig:rredge}
    \end{subfigure}\qquad\hspace{1cm} 
    \begin{subfigure}[b]{0.25\textwidth}
        \includegraphics[width=\textwidth]{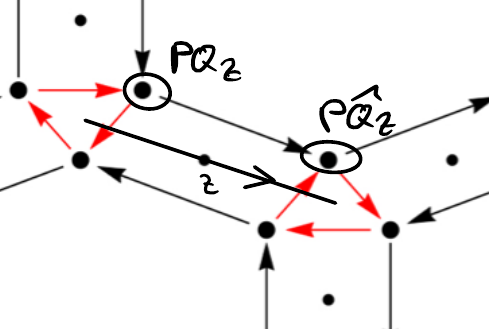}
  \caption{t and t edge}
       \label{fig:ttedge}
        \end{subfigure}\qquad
         \begin{subfigure}[b]{0.20\textwidth}
        \includegraphics[width=\textwidth]{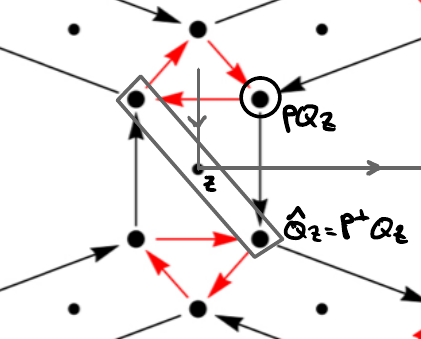}
        \caption{outgoing r and t edge}
       \label{fig:rtedge}
        \end{subfigure}
    \caption{Parts contributing and non contributing to $\ind(\Phi)$}
     \label{fig:randtpath}
\end{figure}
\begin{lem}\label{lem:estimate}Let $\gamma$ be an admissible  $r$-path in $\bR$; then it holds for the operator and the trace norm: 
\[\Vert\Phi_\gamma\Vert\le\sup_{z\in \cK_\gamma}\vert r_z\vert,\quad \Vert\Phi_\gamma\Vert_1\le const. \sum_{z\in \cK_\gamma}\vert r_z\vert.\]

\noindent The analogous statement holds true for a $t$-path
\end{lem}

\begin{proof}  Only states with label in $\cK_\gamma$ contribute,    by Theorem \ref{lem:wandering}.\ref{item:phi2}: $\lbrack \Phi_\gamma^2,Q_z\rbrack=0$  
and $\gamma$ crosses  only  edges of weight $\vert r_z\vert$  thus $\Phi_\gamma^2 Q_z=\vert r_z\vert^2 Q_z$, see figure (\ref{fig:phi2}), so we have
\[\Phi_\gamma^2=PU^\ast P_\gamma^\perp U P_\gamma+P_\gamma^\perp U^\ast P_\gamma U P_\gamma^\perp=\sum_{z\in\cK} \Phi_\gamma^2 Q_z=\sum_{z\in \cK_\gamma} \Phi_\gamma^2 Q_z=\sum_{z\in \cK_\gamma} \vert r_z\vert^2 Q_z.\]. It follows that  $\vert \Phi_\gamma\vert=\sum_{z\in \cK_\gamma}\vert r_z\vert Q_z$ which implies the assertion.

\begin{figure}[hbt]
\centerline {
\includegraphics[width=6cm]{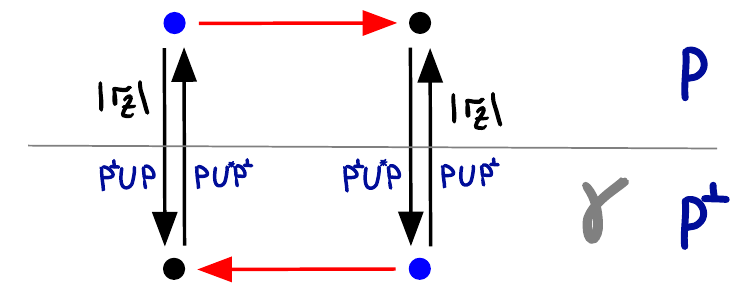}
}
\caption{The action of $\Phi^2Q_z$ for an $r$-path}
\label{fig:phi2}
\end{figure}
\end{proof}

We are now able to prove that a Chalker Coddington model on the Ruby graph which admits a suitable admissible path  defines a projection with non-trivial flux. We show that  a path is suitable if it  crosses only hexagons and rectangles in the past and only triangles and rectangles in the future, and such that the associated weights are sufficiently regular. Remark that an analogous result was proved in Theorem 4.7  in \cite{ABJ5} for the L-network which leads to the original Chalker Coddington model on the associated Manhattan graph. We can adapt this  result without essential difficulty to the present situation.

\begin{thm} \label{thm:main} Let $U$ be a Chalker Coddington model on the Ruby graph such that there exists  an admissible path $\gamma$  which is an $r$-path in $(-\infty, -N\rbrack$ and a  $t$-path in $\lbrack N, \infty)$ for an integer $N>1$, and such that  for a $c>0$ 
\[\vert r_z\vert\le c<1 \quad \forall z\in \cK_{\gamma\restriction (-\infty, -N\rbrack} \hbox{ and } \vert t_z\vert\le c<1 \quad \forall z\in \cK_{\gamma\restriction \lbrack N, \infty)}.\]
Then $\ind(\Phi_\gamma)$ is well defined and it holds
\[\vert\ind(\Phi_\gamma)\vert=1.\]
In addition
\begin{enumerate}
\item $\Phi_\gamma$ is compact iff \quad$\lim_{s\to-\infty}{r_{\gamma(s)}}=0=\lim_{s\to\infty}{t_{\gamma(s)}}$
and then  
\[\sigma(U)=S^1.\]
\item $\Phi_\gamma$ is trace class iff  \quad$\sum_{z\in \cK_{\gamma\restriction (-\infty, -N\rbrack} }\vert r_z\vert + \sum_{z\in \cK_{\gamma\restriction \lbrack N, \infty)}} \vert t_z\vert <\infty$
and then 
\[\sigma_{ac}(U)=S^1,\]
and a trace class perturbation of $U$ contains a shift operator.
\end{enumerate}

\end{thm}
\begin{proof}If $\gamma$ is a path which switches from $r$ to $t$ on one site, i.e : if $\gamma$ is an $r$-path in $(-\infty, 0)$ and a $t$-path in $[0,\infty)$, then as a corollary of Proposition \ref{prop:indexformula} and Remark \ref{rem:dimensions} we have
\[\ind(\Phi)=\sum_{z\in \cK_{\gamma\restriction_{[-N,N]}}}\dim\ran(P\widehat{Q}_z)-\dim\ran(PQ_z)=\dim\ran(P\widehat{Q}_{z_0})-\dim\ran(PQ_{z_0})\]
with  $z_0=\gamma(0)$. By the symmetry of the problem the two edges bisected by the links adjacent to $\gamma(0)$ are either incoming or outgoing to $G_+$ thus  $\dim\ran PQ_{z_0}=1$ and $\dim\ran P\widehat{Q}_{z_0}=0$ or the other way round, see figure (\ref{fig:rtedge}). Thus $\vert \ind(\Phi_\gamma)\vert=1$.

If $\gamma$ switches several times cutting $r$ and $t$ links then  choose two integers $N_\pm\in\bZ$ such that $ \pm N_\pm>\pm N$ and such that $\gamma(N_-)$ is the center of a hexagon and $\gamma(N_+)$ the center of a triangle. Thus all links incident to $\gamma(N_-)$  bisect edges of weight $\vert r\vert$ and all links incident to $\gamma(N_+)$  bisect edges of weight $\vert t\vert$. Then we can define a new admissible path $\widehat\gamma$  replacing the part $\gamma\restriction_{(-N_-,N_+)}$ by an $r$ path inside $G_+$ connecting $\gamma(N_-)$ to $\gamma(N_+)$, c.f. figure (\ref{fig:joinrt}). The new path switches from $r$ to $t$ only at $\gamma(N_+)$.The difference $P_\gamma-P_{\widehat\gamma}$ is of finite rank thus by the invariance properties of the index
 \[\ind(\Phi_\gamma)=\ind(\Phi_{\widehat\gamma}).\]
 The additional assertions follow from Lemma \ref{lem:estimate} and Theorem \ref{thm:general}.
 \begin{figure}
\centerline {
\includegraphics[width=8cm]{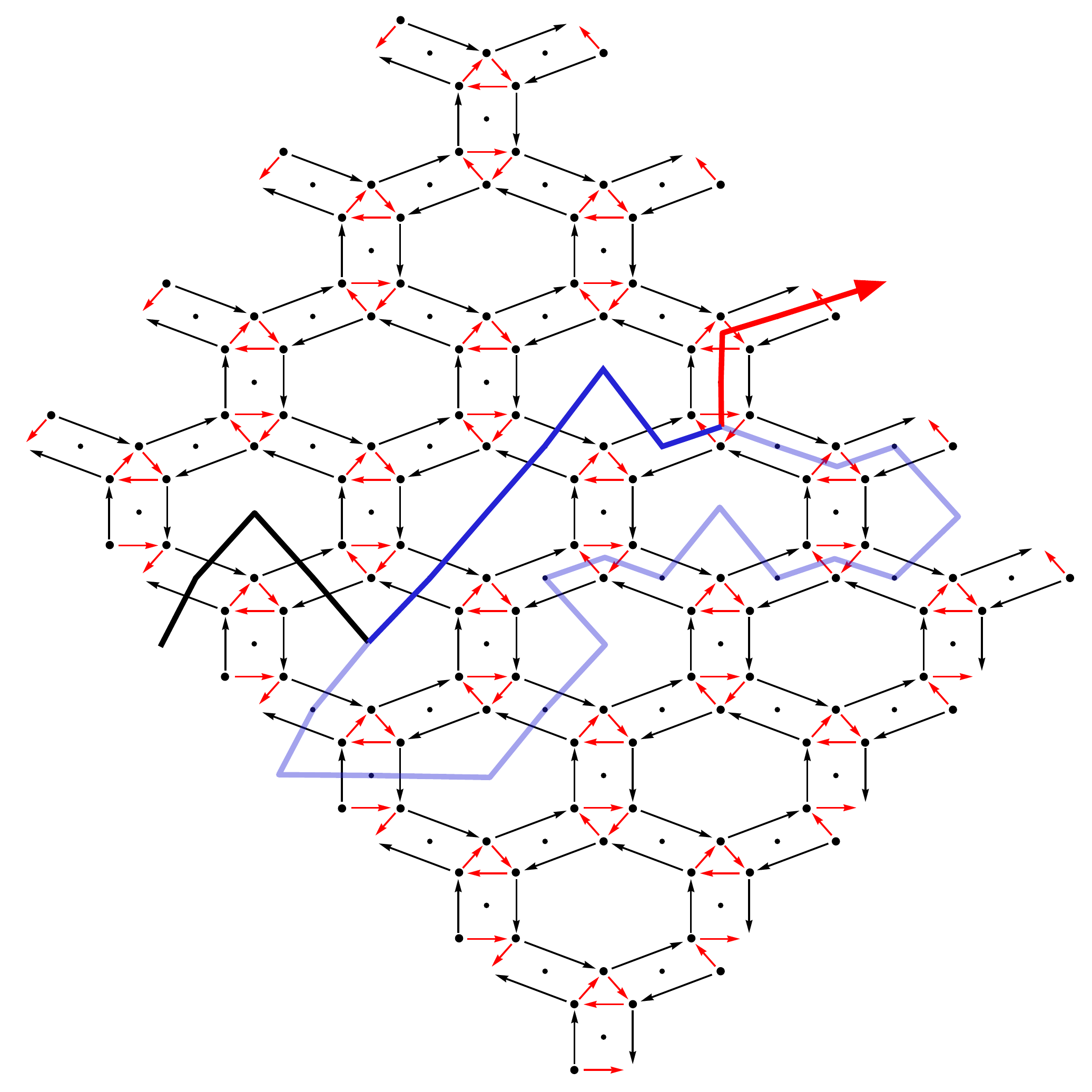}
}
\caption{Join r path to t path in figure (\ref{fig:projection})}
\label{fig:joinrt}
\end{figure}
\end{proof}

As by unitarity we have $\vert r_z\vert^2+\vert t_z\vert^2=1$ it follows:
\begin{cor} 
Any Chalker Coddington model on the Ruby graph defined by a collection of scattering matrices such that  for  $c_1,c_2\in(0,1)$ and all $z$:  $0<c_1<\vert r_z\vert<c_2<1$, admits a projection with non trivial index.
\end{cor}

\begin{rem}\label{rem:cap}A remarkable situation which is covered  by the Corollary  is the translation invariant $0<\vert r_z\vert=const\neq1$ case with all phases equal to unity; here the spectrum is absolutely continuous and exhibits gaps or not.  However, for a $\gamma$ with a $r-t$ transition the flux operator $\Phi_\gamma$ is not trace class or even compact. {The analogous situation is present for perturbations of the translation invariant model, see \cite{ABJ3}}.
\end{rem}

\section{Edge current in two dimensional spin-$1/2$ Quantum Walks}

We analyze transport properties of a class of  two dimensional spin-$1/2$ Quantum Walks proposed in \cite{kitagawa}  for which  transport seems to have been experimentally observed \cite{chenEtAl}. The feature of these walks is that they depend on two $\bC^2$ coins. We use the methods which were developed in \cite{ABJ5} for coined Quantum Walks which depend on one $\bC^4$ coin and analyze the flux out of a subspace defined by a projection which can be interpreted as a quantum lead.

We present the class of models. The Hilbert space is $\bH=\ell^2(\bZ^2;\bC^2)$  its canonical scalar product $\left\langle\cdot, \cdot\right\rangle$.

For a family of $2\times2$ matrices $\{M(z)\}_{z\in\bZ^2}\subset\bM(2;\bC)$ the multiplication operator  $\bbM$ on $\bH$ with symbol $M$ is
\[\bbM\psi(z):=M(z)\psi(z), \quad \psi\in\bH.\]
A coin $\bbC$ is a multiplication operator with unitary symbol $\{C(z)\}_{z\in\bZ^2}\subset U(2)$.

With the canonical basis vectors, $e_x,e_y$ of $\bZ^2$ denote  $S_j\psi(z):=\psi(z-e_j),\quad j\in\{x,y\}$ the shift in the $x$ or $y$ direction. The corresponding conditional shifts are
\[\cshift_j:\bH\to\bH, \quad \cshift_j\psi(z):=
\left(\begin{array}{cc}S_j&   0  \\0  & S_j^\ast \end{array}\right)\psi(z)=\left(\begin{array}{c}\psi_1(z-e_j) \\\psi_2(z+e_j)\end{array}\right).\]
 We use the expressions $\cshift_j=S_j\bbP_+ + S_j^\ast \bbP_-$
with the projections $\bbP_\pm$ whose symbols are $\{P_\pm\}_{z\in\bZ^2}$ the projections on the upper resp. lower component of $\bC^2$. We also use $\bZ^2\ni z=(x, y)$ and denote the canonical basis vectors of $\bH$ by 
\[\Ket{z_0;+}:=\left(\bZ^2\ni z\mapsto \delta_{z,z_0}\left(\begin{array}{c} 1 \\ 0\end{array}\right)\in\bC^2\right),\quad \Ket{z_0;-}:=\left(\bZ^2\ni z\mapsto \delta_{z,z_0}\left(\begin{array}{c} 0 \\ 1\end{array}\right)\in\bC^2\right).\]
Note that $S_j\Ket{z;\pm}=\Ket{z+e_j;\pm}$.

The family of unitary operators which we study depends parametrically on two unitary symbols:

Let $\bbC_1,\bbC_2$ two coins. Define 
\begin{equation}\label{eq:U}
U:\bH\to\bH, \qquad U:=\cshift_y C_2 \cshift_x C_1.
\end{equation}

In order to fix notations, we parametrize the symbols for $z\in\bZ^2$ by
\[C_j(z)=q_j(z)\left(
\begin{array}{cc}
  r_j(z)   &-t_j(z)   \\
  \overline{t_j}(z)  &  \overline{r_j}(z)  
\end{array}
\right)\hbox{ with } q(z)\in S^1, r(z),t(z)\in\bC \hbox{ s.t. } \vert r(z)\vert^2+\vert t(z)\vert^2=1.\]

To illustrate the action of $U$, see figure (\ref{fig:actionofU}), we use  the kernel $U(\cdot,\cdot):\bZ^2\times\bZ^2\to\bM(2;\bC)$ such that $(U\psi)(z)=\sum_{w\in\bZ^2}U(z,w)\psi(w)$ which is
\[U(z,w)=\left\lbrace\begin{array}{ll}
(q_2r_2)(w+e_x)q_1(w)\left(\begin{array}{cc}r_1(w)&-t_1(w)\\0&0\end{array}\right)& \hbox{ if }z=w+e_x+e_y\\
(q_2\overline{r_2})(w-e_x)q_1(w)\left(\begin{array}{cc}0&0\\ \overline{t_1}(w)&\overline{r_1}(w)\end{array}\right)& \hbox{ if }z=w-e_x-e_y\\
(-q_2t_2)(w-e_x)q_1(w)\left(\begin{array}{cc}\overline{t_1}(w)&\overline{r_1}(w)\\0&0\end{array}\right)& \hbox{ if }z=w-e_x+e_y\\
(q_2\overline{t_2})(w+e_x)q_1(w)\left(\begin{array}{cc}0&0\\ {r_1}(w)&-{t_1}(w) \end{array}\right)& \hbox{ if }z=w+e_x-e_y\\
\left(\begin{array}{cc}0&0\\ 0&0 \end{array}\right)& \hbox{ elsewhere }
\end{array}\right.
\]

\begin{figure}[hbt]
\centerline {
\includegraphics[width=5cm]{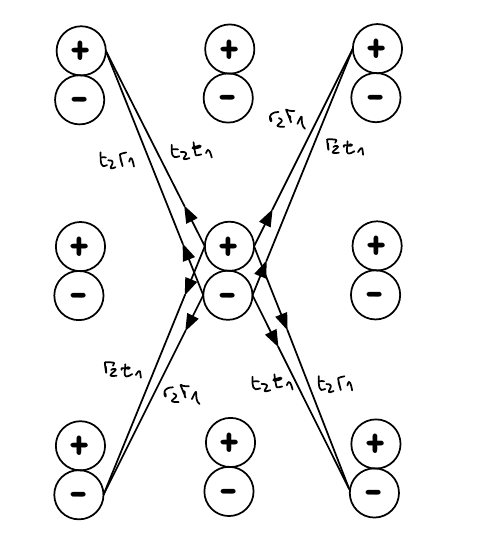}
}
\caption{The action of $U$ with the identifications $r_1\leftrightarrow\vert q_1r_1\vert$ \\ and $r_2\leftrightarrow\vert q_2r_2(\cdot\pm e_x)\vert\ldots$ }
\label{fig:actionofU}
\end{figure}

The basic observation is:

\begin{proposition}\label{propo:bo}
Let $\bbC_1. \bbC_2$ be two coins such that
\[\left.\begin{array}{ll}r_2(x,\phantom{-}0)=0=t_1(x,\phantom{-}0)& x\in\bZ\\t_2(x,-1)=0=r_1(x,-1)& x\in\bZ\end{array}\right. .
\]
Then for $U=\cshift_y\bbC_2\cshift_x\bbC_1$ the subspace
\[\bH_0:=\bigoplus_{x\in\bZ} span\left\{\Ket{x,0;+}, \Ket{x,-1;-}\right\}\]
is invariant for $U$. Furthermore  $U\restriction\bH_0$ is a two dimensional shift with wandering subspace $\bL_0:=span \left\{\Ket{-1,0;+},\Ket{-1,-1;-}\right\}$, i.e. $U^n\bL_0\perp\bL_0, \quad\forall n\neq0$ and $\bigoplus_{n\in\bZ}U^n\bL_0=\bH_0$
\end{proposition}
\begin{proof}The proof is by  explicit calculation which we present for 
\[C_2(x,0)=\left(\begin{array}{cc}0&1\\1&0\end{array}\right) ,\qquad C_1(x,0)=\left(\begin{array}{cc}1&0\\0&1\end{array}\right),
\]
\[C_2(x,-1)=\left(\begin{array}{cc}1&0\\0&1\end{array}\right),\qquad C_1(x,-1)=\left(\begin{array}{cc}0&1\\1&0\end{array}\right),
\]
because for the general case all  the phases which may appear eventually cancel for the projections below. We have for $x\in\bZ$
\[\Ket{x,0;+}\xrightarrow{\bbC_1}\Ket{x,0;+}\xrightarrow{\cshift_x}\Ket{x+1,0;+}\xrightarrow{\bbC_2}\Ket{x+1,0;-}\xrightarrow{\cshift_y}\Ket{x+1,-1;-}\hbox{ and }\]
\[\Ket{x,-1;-}\xrightarrow{\bbC_1}\Ket{x,-1;+}\xrightarrow{\cshift_x}\Ket{x+1,-1;+}\xrightarrow{\bbC_2}\Ket{x+1,-1;+}\xrightarrow{\cshift_y}\Ket{x+1,0;+}\]
and thus it holds for the projections in these directions
\[U\Ket{x,0;+}\bra{x,0;+}U^\ast=\Ket{x+1,-1;-}\Bra{x+1,-1;-}\hbox{ and } \]
\[U\Ket{x,-1;-}\Bra{x,-1;-}U^\ast=\Ket{x+1,0;+}\Bra{x+1,0;+}.\]
By induction and unitarity this implies for $n\in\bZ$
\[U^n\left(\Ket{0,0;+}\bra{0,0;+}+\Ket{0,-1;-}\Bra{0,-1;-}\right){U^\ast}^n=\Ket{n,0;+}\Bra{n,0;+}+\Ket{n,-1;-}\Bra{n,-1;-} \]
which proves the three assertions.
\end{proof}

\begin{remark}Remark that under the conditions of proposition \ref{propo:bo} the absolutely continuous spectrum of $U$ is the full unit circle but other types of spectra of $U$ typically coexist because: $\sigma_{ac}(U\restriction\bH_0)=S^1\subset\sigma_{ac}(U)$ and the coins $\bbC_1,\bbC_2$ are arbitrary  on $\bH_0^\perp$.
\end{remark}

Now making use of the theory of the index of two projections and its topological invariance we can substantially extend this to a precise result on the existence of  a current which was observed in the experiment of \cite{chenEtAl}. We consider the projection on a two state quantum half line, i.e.:
\begin{definition}Let  $\bbP$ be the projection valued multiplication operator with symbol $P$ defined by 
\begin{equation}\label{eq:P}
P(x):=\left\lbrace
\begin{array}{ll}
P_+ & x\ge0 \hbox{ and } y=0\\
P_- &  x\ge0 \hbox{ and } y=-1\\
0 & \hbox{elsewhere}
\end{array}
\right. .
\end{equation}
\end{definition}
It is a corollary of Proposition \ref{propo:bo} that $\bL_0$ is the incoming subspace for $\ran P$:

\begin{remark}$U$ defined in Proposition \ref{propo:bo} is a forward shift in $\ran P$, $U^n\bL_0\perp\bL_0,\forall n\in\bN$ and $\ran P=
\bigoplus_{n\in\bN}U^n\bL_0$.
\end{remark}

This existence of a current holds true for a large class of  $U$ in the following precise sense:

\begin{thm}\label{thm:currentKwalk}
Let $\bbC_1. \bbC_2$ be two coins such that
\[\sum_{x\in\bN}\left(
\left\vert r_2(x,0)\right\vert+\left\vert t_1(x,0)\right\vert+\left\vert t_2(x,-1)\right\vert+\left\vert r_1(x,-1)\right\vert
\right)<\infty
\]
then it holds for $U$ defined in (\ref{eq:U}),  $P$ in (\ref{eq:P}) and $\Phi:=U^\ast P U-P$:
\[\Phi \hbox{ is trace class with }  \ind(\Phi)=2,
\]
\[\sigma_{ac}(U)=S^1,\]
and there exist a unitary trace class perturbation $\widehat{U}=S\oplus\widetilde{U}$ of $U$ with $S$ a bilateral shift of multiplicity $2$; further there exists a $2$ dimensional subspace 
\[\bL\subset\ker P^\perp UP^\perp\restriction_{\ran P^\perp}\] such that 
\[S^n\bL\perp\bL, \forall \bZ\ni n\neq0 \quad \hbox{and}\quad S^n\bL=PS^n\bL,\forall n\in\bN.\]
\end{thm}
\begin{proof}Separating the phases we write for the parameters of $\bbC_1,\bbC_2$: $r_j=e^{i\rho_j}\vert r_j\vert=e^{i\rho_j}\sqrt{1-\vert t_j\vert^2}$ and $t_j=e^{i\tau_j}\vert t_j\vert=e^{i\tau_j}\sqrt{1-\vert r_j\vert^2}$.

Define two coins $\bbC_1^0,\bbC_2^0$ by the symbols

\[C_2^0(z):=\left\lbrace
\begin{array}{ll}
q_2\left(\begin{array}{cc}0&-e^{i\tau_2}\\e^{-i\tau_2}&0\end{array}\right)(z)& \hbox{ if }z=(x,0),\forall x\in\bN\\
q_2\left(\begin{array}{cc}e^{i\rho_2}&0\\0&e^{-i\rho_2}\end{array}\right)(z)& \hbox{ if }z=(x,-1),\forall x\in\bN\\
C_2(z)& \hbox{ elsewhere }
\end{array} ,
\right.\]
\[C_1^0(z):=\left\lbrace
\begin{array}{ll}
q_1\left(\begin{array}{cc}e^{i\rho_1}&0\\0&e^{-i\rho_1}\end{array}\right)(z)& \hbox{ if }z=(x,0),\forall x\in\bN\\
q_1\left(\begin{array}{cc}0&-e^{i\tau_1}\\e^{-i\tau_1}&0\end{array}\right)(z)& \hbox{ if }z=(x,-1),\forall x\in\bN\\
C_1(z)& \hbox{ elsewhere }
\end{array} .
\right.\]

For the Hilbert-Schmidt norm $\Vert\cdot\Vert_{HS}$ on $\bM(2;\bC)$ we obtain 
\[\left\Vert (C_2-C_2^0)(z)\right\Vert_{HS}^2\le 4
\left\{\begin{array}{ll} \vert r_2(z)\vert^2& z=(x,0),\forall x\in\bN \\ \vert t_2(z)\vert^2& z=(x,-1), \forall x\in\bN\\0& \hbox{elsewhere} \end{array}\right. ,
\]
\[\left\Vert (C_1-C_1^0)(z)\right\Vert_{HS}^2\le 4
\left\{\begin{array}{ll} \vert t_1(z)\vert^2& z=(x,0),\forall x\in\bN \\ \vert r_1(z)\vert^2& z=(x,-1), \forall x\in\bN\\0& \hbox{elsewhere}\end{array}\right. .
\]
It follows for the trace norm
\begin{eqnarray*}
&&\Vert \bbC_2-\bbC_2^0\Vert_1+\Vert \bbC_1-\bbC_1^0\Vert_1 \le const. \sum_{z\in\bZ^2}\sum_j \Vert (C_j-C_j^0)(z)\Vert_{HS}\le
\\
&&\sum_{x\in\bN}\left(
\left\vert r_1(x,-1)\right\vert+\left\vert t_2(x,-1)\right\vert+\left\vert t_1(x,0)\right\vert+\left\vert r_2(x,0)\right\vert
\right)<\infty.
\end{eqnarray*}
For $U^0:=\cshift_y\bbC_2^0\cshift_x\bbC_1^0$, $\Phi^0:={U^0}^\ast P U^0-P$ it follows that the difference 
\[U-U^0=\cshift_y(\bbC_2-\bbC_2^0)\cshift_x\bbC_1+\cshift_y\bbC_2\cshift_x(\bbC_1-\bbC_1^0)\hbox{ is trace class }\]
and thus
\[\Phi-\Phi^0\hbox{ is trace class }.\]
$\Phi^0$ equals minus the orthogonal projection on  $\bL_0$. It follows that $\Phi$ is trace class and that $\ind(\Phi)=2$. 

Now Theorem \ref{thm:general}.\ref{item:trace} and the construction of $\bL$ in its proof (see \cite{ABJ5} Theorem 2.1, adapted by replacing $P\leftrightarrow P^\perp$) implies the assertion.
\end{proof}

As an illustration of this result of the existence of a  current in the subspace $\ran P$ for a trace class perturbation of $U$, we explicit the non-triviality of the incoming subspace $\bL$  for $U$ for a  finite rank perturbation of the paradigmatic situation of Proposition \ref{propo:bo}.

\begin{remark}\label{rem:illustration}
Let $\bbC_1. \bbC_2$ be two coins such that, c.f. figure (\ref{fig:illustration})
\[\left.\begin{array}{ll}r_2(x,\phantom{-}0)=0=t_1(x,\phantom{-}0)& x\in\bN\setminus\{1,2\}\\t_2(x,-1)=0=r_1(x,-1)& x\in\bN\bigcup\{-1,0\}\end{array}\right. .
\]
Then for the incoming state $\bL$ of Theorem \ref{thm:currentKwalk} it holds 
 $\bL:=span \left\{\Ket{-1,-1;-},\psi\right\}$ where $\psi$ has generically non zero components on all vectors in $\{\Ket{-1,0;+},\Ket{1,0;+},\Ket{3,0;+}\}$
\end{remark}

\begin{figure}[hbt]
\centerline {
\includegraphics[width=5cm]{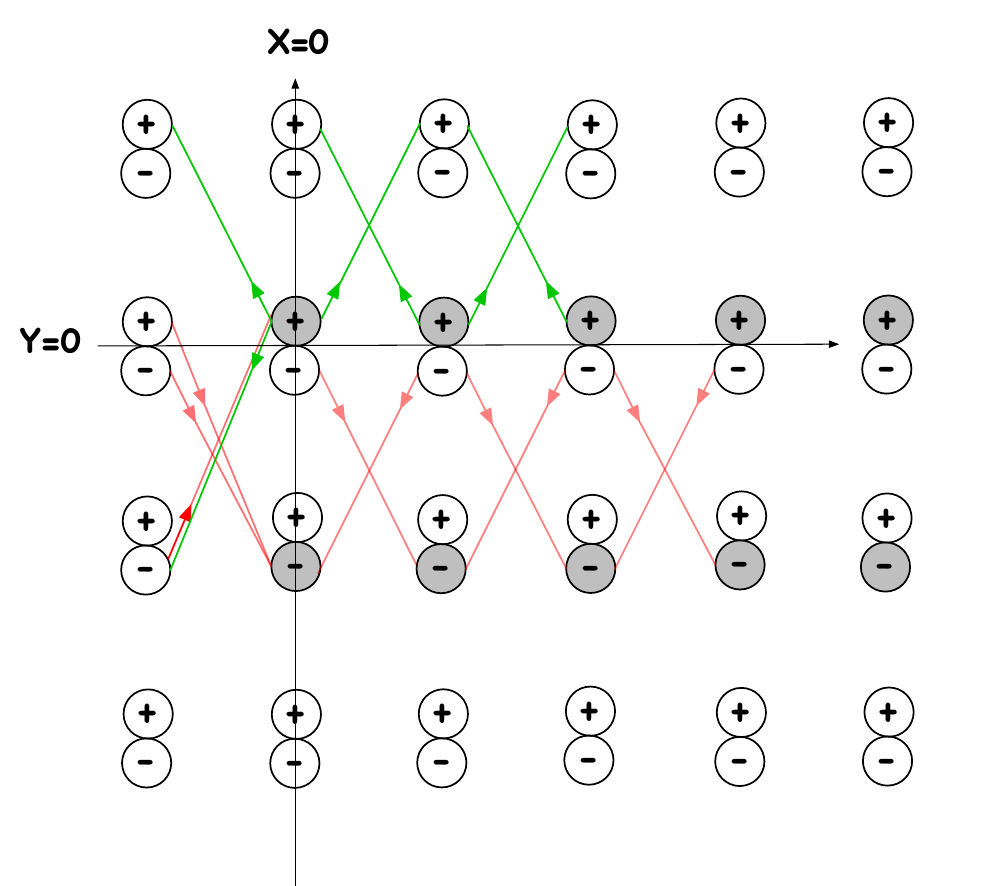}
}
\caption{Illustration of the incoming space for $U$ in Remark (\ref{rem:illustration}). The gray dots represent the states in $\ran P$, green arrows the elements of $P^\perp U P$, red arrows the incoming elements of $PUP^\perp$}
\label{fig:illustration}
\end{figure}

\appendix
\section{General results on the flux operator $\Phi$}
The important notion of the relative index of  two projections was defined in \cite{ASS}. We list here some of its known properties which are used above. For the proofs we refer to \cite{ASS} and \cite{ABJ5} as concerns the spectral and dynamical implications.
\begin{thm}\label{thm:general}
Let $U$ be a unitary operator on a Hilbert space and $P$ an orthogonal projection. For the selfadjoint operator
\[\Phi:=U^\ast P U-P=U^\ast\left[P,U\right] \] 
suppose that  $1$ is an isolated eigenvalue of finite multiplicity of  $\Phi^2$ and define the integer
\[\ind(\Phi)
:= \dim\ker\left(\Phi-\bI\right)-\dim\ker\left(\Phi+\bI\right).\]
If the index does not vanish,  $\ind(\Phi)=n\neq0$, then:

there exists a unitary $\widehat{U}$ such that   $\widehat{U}=S\oplus\widetilde{U}$ where $S$ is a bilateral shift of multiplicity $\vert n\vert$, $\widetilde{U}$ is unitary on its subspace , $\lbrack\widetilde{U},P\rbrack=0$ and :
\begin{enumerate}
\item $\Vert U-\widehat{U}-F\Vert=\cO(\Vert\Phi_<\Vert)$ for a finite rank operator $F$  and $\Phi_<$ the restriction of $\Phi$ to its spectral subspace off $\pm1$ : $\Phi_<:=\Phi\chi(\Phi^2<1)$;
\item
if  $\lbrack P,U\rbrack$ is compact then $U-\widehat{U}$ is compact   and  the essential spectrum of $U$ is the whole unit circle :
\[\sigma(U)=S^1;\]
\item\label{item:trace} if  $\lbrack P,U\rbrack$ is trace class then $U-\widehat{U}$ is trace class   and the absolutely continuous spectrum of $U$ is the whole unit circle:
\[\sigma_{ac}(U)=S^1.\]
\end{enumerate}\end{thm}

\begin{proposition}{\label{lem:wandering}}With $\Phi$ as in Theorem \ref{thm:general} it holds:
\begin{enumerate}
\item \label{item:phi2}\[\Phi^2\le1; \quad\lbrack\Phi^2,P\rbrack=0;\]
\[ \ker\left(\Phi+\bI\right)=\ker\left(PUP\restriction\ran{P} \right); \quad\ker\left(\Phi-\bI \right)=\ker\left( P^\perp U P^\perp\restriction\ran{P^\perp}\right).\]
\item If $\ind(\Phi)$ is defined then  $PUP$ is Fredholm on $\ran P$ and $\ind(\Phi)$ equals minus its Fredholm index :
\[\ind(\Phi)=\dim\ker QU^\ast Q-\dim\ker QUQ.\]
\item  \[\ind(\Phi)= \dim\ker\left((\Phi^2-\bI)\restriction\ran{P^\perp}\right)-\dim\ker\left((\Phi^2-\bI)\restriction\ran{P}\right).\]

\item If $\lbrack0,1\rbrack\ni t\to U(t)$ is norm continuous and unitary and for $\Phi(t)=U^\ast(t)PU(t)$: $1\notin\sigma_{ess}(\Phi(t)^2)$ then $\bZ\ni\ind(\Phi(t))=const.$ 
\item\label{item:compact} For unitaries $U_0$, $U_1$ such that $U_1-U_0$ is a compact operator, it holds for the corresponding flux operators $\Phi_0,\Phi_1$:\label{compact}
\[\ind(\Phi_0)=\ind(\Phi_1).\]
\item  A $d$-dimensional subspace $\bL$ is called wandering  for a unitary $U$, if ${{U}}^k\bL\perp\bL\quad \forall k\in\bN$. 
For an orthogonal decomposition $\bL=\bigoplus_{j=1}^{d}\bL_j$ into 1-dimensional subspaces  and for the $U$-invariant subspace
\[\bM:=\bigoplus_{k\in\bZ} {{U}}^k\bL=\bigoplus_{j=1}^{d} \bigoplus_{k\in\bZ}{{U}}^k\bL_j\]
it holds that $S:={U}\restriction\bM$ is a bilateral shift of multiplicity  $d$   and $U\restriction\bM^\perp$ is unitary on $\bM^\perp$. In particular $\sigma_{ac}(U)=\bS^1$.
\end{enumerate}
\end{proposition}

{\bf\large  Acknowledgments} \vspace{.2cm}

This article is part of the requirements for the  doctoral thesis of Mohamed Mouneime.

\end{document}